\documentclass[conference,hidelinks]{IEEEtran}

\usepackage{amsmath,amssymb,amsthm,bm}
\usepackage{graphicx}
\usepackage{cite}
\usepackage{booktabs}
\usepackage{algorithm}
\usepackage{algorithmic}
\usepackage{lipsum}
\usepackage{xcolor}
\usepackage{hyperref}
\usepackage[acronym,shortcuts]{glossaries}

\newcommand{\C}{\mathbb{C}}
\newcommand{\R}{\mathbb{R}}
\newcommand{\Us}{\mathcal{U}_{\mathrm{s}}}
\newcommand{\Chat}{\widehat{\mathcal{C}}_{\mathrm{gc}}}
\newcommand{\Mhat}{\widehat{\mathcal{M}}_{\mathrm{gc}}}
\newcommand{\Mgc}{\mathcal{M}_{\mathrm{gc}}}
\newcommand{\Ig}{\mathcal{I}_g}
\newcommand{\mA}{\mathbf{A}}

\newcommand{\mE}{\mathbf{E}}
\newcommand{\mHd}{\mathbf{H}_{\mathrm{d}}}
\newcommand{\mHrx}{{\mathbf{H}_{\mathrm{RX}}}}
\newcommand{\mHrxg}{{\mathbf{H}_{\mathrm{RX},g}}}
\newcommand{\mHtx}{{\mathbf{H}_{\mathrm{TX}}}}
\newcommand{\mHtxg}{{\mathbf{H}_{\mathrm{TX},g}}}
\newcommand{\mI}{\mathbf{I}}
\newcommand{\mJ}{\mathbf{J}}

\newcommand{\mK}{\mathbf{K}}
\newcommand{\mO}{\mathbf{O}}
\newcommand{\mQ}{\mathbf{Q}}
\newcommand{\mR}{\mathbf{R}}
\newcommand{\mS}{\mathbf{S}}
\newcommand{\mU}{\mathbf{U}}
\newcommand{\mV}{\mathbf{V}}
\newcommand{\mZ}{\mathbf{Z}}
\newcommand{\bTheta}{\bm{\Theta}}
\newcommand{\bxi}{\bm{\xi}}
\newcommand{\bzetahat}{\widehat{\bm{\zeta}}}
\newcommand{\tr}{\operatorname{tr}}
\newcommand{\diag}{\operatorname{diag}}
\newcommand{\blkdiag}{\operatorname{blkdiag}}

\newcommand{\Tmap}{\mathcal{T}}

\theoremstyle{definition}

\theoremstyle{plain}
\newtheorem{proposition}{Proposition}
\theoremstyle{remark}
\newtheorem{remark}{Remark}

\title{Group-Connected Riemannian Manifold Optimization for Reciprocal BD-RIS}

\author{\IEEEauthorblockN{Marko Fidanovski$^{\star}$, Thushar Venkataramanaiah$^{\star}$, {Eduard Jorswieck$^{\dag}$, and} Giuseppe Thadeu Freitas de Abreu$^{\star}$}
\IEEEauthorblockA{$^{\star}$\textit{School of Computer Science and Engineering, Constructor University, 28759 Bremen, Germany}\\
$^{\dag}$\textit{Institute of Communications Technology, TU Braunschweig, 38106 Braunschweig, Germany}\\
Emails: [mfidanovsk, tvenkatara, gabreu]@constructor.university, {e.jorswieck@tu-braunschweig.de}}
}

\begin{document}

\newacronym{RIS}{RIS}{reconfigurable intelligent surface}
\newacronym{BD-RIS}{BD-RIS}{beyond-diagonal RIS}
\newacronym{MIMO}{MIMO}{multiple-input multiple-output}
\newacronym{MO}{MO}{manifold optimization}
\newacronym{PO}{PO}{phase optimization}
\newacronym{FP}{FP}{fractional programming}

\newacronym{iff}{iff}{if and only if}

\maketitle

{

\begin{abstract}
\Acp{RIS} provide a flexible means of engineering wireless propagation channels through configurable multiport scattering networks.
Among their architectures, group-connected \acp{BD-RIS} offer a practical tradeoff between scattering flexibility and hardware complexity by partitioning the surface into interconnected groups.
In reciprocal and lossless implementations, each group scattering block must be symmetric and unitary.
In this paper, we investigate sum-rate maximization for reciprocal group-connected \ac{BD-RIS}-assisted \ac{MIMO} systems under such structural constraints.
In particular, we build on a recent result where the feasible space determined by the symmetry and unitary constraints was characterized as a Riemannian manifold, extending it into a product manifold of symmetric-unitary block manifolds, thus expanding the approach to manifold phase optimization in the group-connected setting. 
As a consequence, the proposed group-extended framework performs tangent-space projections and Takagi retractions blockwise, while preserving the coupling among all groups through the common \ac{MIMO} equivalent channel. 
The formulation naturally includes single-connected and fully connected reciprocal \ac{BD-RIS} architectures as special cases and enables a flexible performance-complexity tradeoff through the group size.
\end{abstract}



\glsresetall
\glslocalunset{RIS}

\vspace{-2ex}
\section{Introduction}
\vspace{-2ex}

\IEEEPARstart{R}{econfigurable} intelligent surfaces (RISs) have emerged as a promising technology for future wireless communication systems, enabling the wireless propagation environment to be engineered rather than treated as a fixed medium. 
By controlling the response of many passive reconfigurable elements, \acp{RIS} can improve coverage, strengthen desired links, suppress interference, and enhance energy efficiency without requiring a proportional increase in active radio-frequency chains \cite{DiRenzo2020,Shen2022Scattering}. 
Early \ac{RIS} models mainly considered diagonal scattering matrices, where each element applies an independent phase shift. 
This single-connected architecture is simple and attractive for implementation, but its diagonal structure limits the available wave-domain processing capability \cite{Shen2022Scattering,Li2023BeyondDiagonal}.

\Acp{BD-RIS} address this limitation by allowing controllable interconnections among the surface elements \cite{Shen2022Scattering,Li2023BeyondDiagonal}. 
Scattering-parameter network models have shown that group-connected and fully connected impedance-network architectures can manipulate incident waves more generally than conventional single-connected \acp{RIS} \cite{Shen2022Scattering}. 
The architecture hierarchy was further developed in \cite{Li2023BeyondDiagonal}, where single-, group-, and fully-connected \acp{BD-RIS} were studied as increasingly expressive scattering models. 

In this hierarchy, the group-connected case is especially important because it provides an intermediate design, \emph{i.e.}, it offers richer scattering behavior than a diagonal \ac{RIS} while avoiding the full hardware complexity of complete interconnection. 
The fundamental motivation for this intermediate design was made precise in \cite{Nerini2023Pareto}, where the Pareto frontier of the performance--complexity tradeoff for \ac{BD-RIS} was derived, showing that single-connected and fully connected architectures are themselves the two extreme special cases of the broader group-connected family. 
Around the same time, \cite{Kim2023GroupConnected} considered scattering-matrix design specifically for a group-connected impedance network in a multiuser \ac{MIMO} setting, reinforcing that group-connected \ac{BD-RIS} is of direct interest for practical, hardware-constrained system design rather than only a theoretical interpolation between the two extremes.

Early \ac{BD-RIS} optimization studies often considered general beyond-diagonal scattering matrices \cite{Fang2024LowComplexityBDRIS,Zhou2024UnifiedBDRIS}, which may also represent non-reciprocal architectures and therefore provide a large number of design degrees of freedom.
However, non-reciprocal implementations generally require additional active or non-reciprocal circuit components, increasing hardware complexity \cite{Fidanovski2025Part1,Li2024NonReciprocalBDRIS}.
This has motivated the study of reciprocal lossless \ac{BD-RIS} architectures, characterized by symmetric and unitary scattering matrices \cite{Shen2022Scattering,Li2023BeyondDiagonal}.

Handling these coupled structural constraints has led to different optimization approaches.
Some works enforce them indirectly, including penalty dual decomposition \cite{Zhou2024UnifiedBDRIS} and convex relaxation followed by symmetric-unitary projection \cite{Fang2024LowComplexityBDRIS}.
Although these approaches yield practical optimization procedures, they do not optimize directly over the manifold structure induced by the simultaneous symmetry and unitarity constraints.
More recent works have explicitly formulated reciprocal \ac{BD-RIS} design as a \ac{MO} problem.
{In \cite{Fidanovski2025Part1,Fidanovski2026FPManifold,MoralesSandoval2026Distributed} unitarity is handled} through optimization over the Stiefel manifold, and {reciprocity is promoted} through a quadratic regularization term in the \ac{FP}-reformulated quadratic objective; {the most recent work extends the framework to distributed \ac{MIMO} beamforming.}

{More broadly, \ac{MO} is well suited to wireless variables with constant-modulus, orthonormal, or unitary structure as tangent-space updates and retractions exploit their intrinsic geometry while preserving feasibility and reducing reliance on penalty tuning or post hoc projection \cite{deSouza2024ManifoldRIS,Geng2024UnifiedManifold}.
Recent works also apply \ac{MO} to multicell interference management and joint \ac{RIS}--precoder design \cite{PanTWC2020,ZhongTVT2025}.
Its demonstrated use in \ac{RIS}-aided \ac{MIMO} and sensing waveform design \cite{deSouza2024ManifoldRIS,Geng2024UnifiedManifold,Rexhepi2025PAPR} establishes a strong foundation for scalable, hardware-aware \ac{BD-RIS} optimization and future quantum extensions \cite{Rexhepi2025QuantumMO}.}

A related work on \ac{BD-RIS}-assisted \ac{MIMO} physical-layer security introduces auxiliary variables to decouple the symmetry and unitarity constraints and combines an augmented-Lagrangian formulation with product Riemannian manifold optimization for joint transmit-beamforming and \ac{BD-RIS} coefficient design \cite{Xiong2025PLSBDRIS}.
These formulations exploit manifold optimization but do not directly optimize over the exact feasible set defined by the simultaneous symmetry and unitarity constraints.

In contrast, Santamaria \emph{et al.} \cite{Santamaria2026ICASSP,Santamaria2026SymmetricUnitary} characterize the intersection of the corresponding constraints as a smooth Riemannian manifold and develop phase-optimization methods that preserve both constraints throughout the iterations.
This geometric formulation is closely related to classical Takagi factorization \cite{Takagi1924Factorization} and the broader framework of Riemannian optimization on matrix manifolds \cite{Absil2008Optimization,Boumal2023SmoothManifolds}.
Extending this manifold-based treatment to the group-connected setting, this work considers reciprocal lossless group-connected \ac{BD-RIS} scattering matrices and formulates their exact feasible space as a product of symmetric-unitary block manifolds.
The resulting blockwise Takagi-based \ac{MO} framework respects both the product-manifold structure and the block-diagonal hardware architecture while retaining the coupling among groups via an equivalent \ac{MIMO} channel.



More specifically, the symmetric-unitary manifold methods in \cite{Santamaria2026ICASSP,Santamaria2026SymmetricUnitary} are developed for fully connected reciprocal \acp{BD-RIS}, where all reconfigurable elements are coupled through a single dense scattering matrix.
Consequently, their formulation does not directly capture the block-diagonal structure of group-connected architectures, in which each group is associated with a separate symmetric-unitary scattering block.
This structural distinction prevents a direct application of the fully connected modeling approach to the group-connected setting.



The main contribution of this paper is a product-manifold reformulation and optimization framework for reciprocal group-connected \ac{BD-RIS} scattering-matrix design.
Specifically, we construct the feasible configuration space as the product of per-group symmetric-unitary manifolds, establishing its smooth product-manifold structure with a corresponding block-diagonal realization.
Building on this geometry, we derive the product tangent space and metric, blockwise tangent projections, and Takagi-based retractions, incorporating them into a product \ac{MO}-based phase-optimization algorithm for sum-rate maximization.
The resulting framework preserves the group-connected hardware structure while retaining the global coupling among the scattering blocks.
Moreover, it naturally recovers the single-connected and fully connected reciprocal \ac{BD-RIS} architectures as special cases, with the group size governing the associated performance-complexity tradeoff.

}

\vspace{-0.5ex}
\section{System Model}

We consider a reciprocal lossless \ac{BD-RIS} with \(R\) reconfigurable elements assisting an \(N_t\times N_r\) \ac{MIMO} link.
The direct channel is \(\mHd\in\C^{N_r\times N_t}\), the \ac{BD-RIS}-to-receiver channel is \(\mHrx\in\C^{N_r\times {R}}\), and the transmitter-to-\ac{BD-RIS} channel is \(\mHtx\in\C^{{R}\times N_t}\).
For a scattering matrix \(\bTheta\in\C^{R\times R}\), the equivalent channel is denoted by
\vspace{-1ex}
\begin{equation}
\mE=\mHd+\mHrx\bTheta\mHtx.
    \label{eq:equivalent_channel_full}
    \vspace{-1ex}
\end{equation}

The achievable rate for isotropic signaling with total transmit power \(P\) and noise variance \(\sigma^2\) is
\vspace{-1ex}
\begin{equation}
    \eta
    =
    \log_2\det\!\Big(
    \mI_{N_r}
    +
    \tfrac{P}{{N_t}\sigma^2}\mE\mE^H
    \Big).
    \label{eq:rate_objective}
\vspace{-0.5ex}
\end{equation}

In turn, the passive reciprocal scattering constraints follow the model of reciprocal \acp{RIS} with the single-, group-, and fully connected \ac{BD-RIS} architecture hierarchy of \cite{Shen2022Scattering,Li2023BeyondDiagonal}.

Let the elements be partitioned into \(G\) disjoint groups, each with $R_G = R \slash G$ reconfigurable elements.
Then, the \(g\)th scattering block is given by
\vspace{-1ex}
\begin{gather}
    \bTheta_g
    =
    \bTheta_{[\Ig,\Ig]}\in\C^{R_G\times R_G},
    \label{eq:group_scattering_block}
    \\
    \Ig
    =
    \{R_G(g-1)+1:gR_G\},
    \vspace{-1ex}
\end{gather}
where \(g\in\{1,\ldots,G\}\).

The corresponding channel matrices are
\vspace{-1ex}
\begin{equation}
\begin{aligned}
    \mHrxg
    &=
    \mHrx_{[:,\Ig]}\in\C^{N_{r}\times R_{G}},
    \\
    \mHtxg
    &=
    \mHtx_{[\Ig,:]}\in\C^{R_{G}\times N_{t}}.
\end{aligned}
    \label{eq:group_channel_blocks}
    \vspace{-1ex}
\end{equation}

The cases \(R_G=1\) and \(R_G=R\) recover the single-connected and fully connected architectures, respectively.

\section{From Group Sets to a Product Manifold}
\vspace{-0.5ex}
The \(g\)th reciprocal lossless scattering block is a complex matrix \(\bTheta_g\in\C^{R_G\times R_G}\) that is symmetric due to reciprocity and unitary due to losslessness.
Its admissible set is
\begin{equation}
    \Us(R_G)
    \!=\!
    \left\{
    \bTheta_g\in\C^{R_G\times R_G}\!:\!
    \bTheta_g\!=\!\bTheta_g^T,\;
    \bTheta_g\bTheta_g^H\!=\!\mI_{R_G}
    \right\}.
    \label{eq:block_su_set}
\end{equation}

At this point,{ we look at} \(\Us(R_G)\) only {as a} feasible set, {ignoring the topology}.
The smooth manifold structure is introduced after the single-block result is recalled.
\vspace{-0.5ex}
\begin{proposition}[Cartesian product of admissible group choices]
The abstract set of feasible group-connected configurations, where each group has the same block size $R_G$, is
\vspace{-1.75ex}
\begin{equation}
\begin{aligned}
    \Chat
    &=
    \overbrace{\Us(R_G) \times \Us(R_G) \times \dots \times \Us(R_G)}^{G\text{ factors}}
    \\
    &=
    \left\{
    (\bTheta_1,\ldots,\bTheta_G):
    \bTheta_g\in\Us(R_G),\;g=1,\ldots,G
    \right\}.
\end{aligned}
    \label{eq:cartesian_product_set}
\end{equation}
\end{proposition}
\begin{proof}
\vspace{-0.5ex}
A feasible group-connected configuration contains one scattering block per group.
By \eqref{eq:block_su_set}, the \(g\)th block is admissible \ac{iff} \(\bTheta_g\in\Us(R_G)\).
The finite Cartesian product is defined as the set of ordered tuples whose \(g\)th component lies in the \(g\)th factor set.
Thus, a tuple belongs to the feasible configuration set \ac{iff} it belongs to the product in \eqref{eq:cartesian_product_set}.

\vspace{-2.75ex}
\end{proof}
\vspace{-1ex}
The physical scattering matrix is obtained only after the tuple has been chosen.
It is the block-diagonal matrix
\vspace{-0.5ex}
\begin{equation}
    \bTheta
    =
    \blkdiag(\bTheta_1,\ldots,\bTheta_G).
    \label{eq:block_diagonal_scattering}
\vspace{-0.5ex}
\end{equation}

Substituting \eqref{eq:block_diagonal_scattering} into \eqref{eq:equivalent_channel_full} yields the group-wise equivalent channel
\vspace{-0.5ex}
\begin{equation}
    \mE
    =
    \mHd+\sum_{g=1}^{G}\mHrxg\bTheta_g\mHtxg .
    \label{eq:group_equivalent_channel}
\vspace{-0.5ex}
\end{equation}

\section{Smooth Structure and Coordinate Maps}
\label{sec:smooth_structure}

The set in \eqref{eq:block_su_set} is not claimed to be smooth merely because it is written as an intersection of constraints.
The required smoothness follows from the characterization of the symmetric-unitary manifold in \cite{Santamaria2026SymmetricUnitary,Santamaria2026ICASSP}.
The key algebraic step is the Takagi factorization: every complex symmetric matrix \(\mA=\mA^T\) can be written as
\vspace{-0.5ex}
\begin{equation}
    \mA=\mQ\bm{\Sigma}\mQ^T,
    \label{eq:takagi_factorization}
\vspace{-0.5ex}
\end{equation}
where \(\mQ\) is unitary and \(\bm{\Sigma}\) is real, diagonal, and nonnegative \cite{Takagi1924Factorization}.
\vspace{-0.5ex}
\begin{proposition}[Single-block smooth manifold]
The set \(\Us(n)\) is a smooth real manifold with dimension \(n(n+1)/2\).
\end{proposition}
\begin{proof}
If \(\mU\in\Us(n)\), then all Takagi singular values in \eqref{eq:takagi_factorization} are equal to one, and
\vspace{-1ex}
\begin{equation}
    \mU=\mQ\mQ^T .
    \label{eq:symmetric_unitary_takagi}
\vspace{-0.5ex}
\end{equation}

Thus the map \(\pi:\mathrm{U}(n)\rightarrow\Us(n)\), \(\pi(\mQ)=\mQ\mQ^T\), is surjective.
For any real orthogonal matrix \(\mO\in\mathrm{O}(n)\), \(\pi(\mQ\mO)=\pi(\mQ)\).
Conversely, equal images differ by such a real orthogonal factor.
Hence \(\Us(n)\) is identified with the quotient \(\mathrm{U}(n)/\mathrm{O}(n)\), {where, for completeness, \(\mathrm{U}(n)\) and \(\mathrm{O}(n)\) denote unitary and real orthogonal groups, respectively}.
Since \(\mathrm{U}(n)\) is a smooth Lie group and \(\mathrm{O}(n)\) is a closed Lie subgroup, the quotient is a smooth manifold.
Its real dimension is
\vspace{-0.5ex}
\begin{equation}
    \dim_{\R}\Us(n)
    =
    n^2-\frac{n(n-1)}{2}
    =
    \frac{n(n+1)}{2}.
    \label{eq:single_dimension}
\vspace{-1ex}
\end{equation}
\vspace{-2ex}
\end{proof}

After each factor is endowed with this smooth structure, the same underlying Cartesian-product set in \eqref{eq:cartesian_product_set} becomes the product manifold
\vspace{-2ex}
\begin{equation}
    \Mhat
    =
    \overbrace{\Us(R_G)\times\cdots\times\Us(R_G)}^{G\ \mathrm{factors}}.
    \label{eq:product_manifold}
    \vspace{-0.5ex}
\end{equation}

The notation \(\Mhat\) emphasizes that the object is no longer only a set of tuples.
It is the tuple set equipped with the product topology and the product smooth {structure}. 
The product manifold, its tangent bundle, Riemannian metric, and local coordinate charts follow directly from the standard construction of product manifolds; see \cite[Ch. 3]{Boumal2023SmoothManifolds} for details.
\vspace{-0.5ex}
\begin{proposition}[Local Euclidean coordinates of the product manifold]
Let \(\widehat{\bTheta}=(\bTheta_1,\ldots,\bTheta_G)\in\Mhat\).
For each factor manifold \(\Us(R_G)\), choose a chart \((\Omega_g,\varphi_g)\) with \(\bTheta_g \in \Omega_g\), where
\vspace{-1ex}
\begin{gather}
    \varphi_g:\Omega_g\subset\Us(R_G)
    \rightarrow
    \varphi_g(\Omega_g)\subset\R^{d_g},\\
    d_g=\frac{R_G(R_G+1)}{2}.
    \label{eq:factor_chart}
    \vspace{-1ex}
\end{gather}
\vspace{-1ex}
Then
\vspace{-1ex}
\begin{gather}
    \varphi
    : \Omega_1\times\cdots\times\Omega_G
    \rightarrow
    \R^{d_1+\cdots+d_G},
    \\
    \varphi(\bTheta_1,\ldots,\bTheta_G)
    =
    \left[
    \varphi_1(\bTheta_1)^T,\ldots,
    \varphi_G(\bTheta_G)^T
    \right]^T,
        \label{eq:product_chart}
\vspace{-1ex}
\end{gather}
is a chart around \(\widehat{\bTheta}\), and consequently,
\vspace{-0.5ex}
\begin{equation}
    \dim_{\R}\Mhat
    =
    G\frac{R_G(R_G+1)}{2}.
    \label{eq:product_dimension}
\vspace{-0.5ex}
\end{equation}
\end{proposition}
\begin{proof}
The product topology has basic neighborhoods of the form \(\Omega_1\times\cdots\times\Omega_G\), with each \(\Omega_g\) open in its factor manifold.
The map in \eqref{eq:product_chart} assigns Euclidean coordinates to points in this neighborhood by stacking the factor-coordinate vectors.
On overlaps, the transition map is the product of the factor transition maps \(\varphi_g\circ\psi_g^{-1}\), and is therefore smooth, as detailed in Appendix~\ref{app:smooth}.
Thus the product charts define a smooth atlas, and the coordinate dimension is the sum of the factor dimensions.
\end{proof}

\begin{remark}[Coordinate map versus physical embedding]
The chart in \eqref{eq:product_chart} is the local map from a neighborhood of a manifold point to a Euclidean coordinate space.
It should not be confused with the block-diagonal physical embedding
\begin{gather}
    \iota : \Mhat\rightarrow\C^{R\times R}\cong\R^{2R^2},\\
    \iota(\bTheta_1,\ldots,\bTheta_G)
    =
    \blkdiag(\bTheta_1,\ldots,\bTheta_G).
    \label{eq:physical_embedding}
\end{gather}

The map \(\iota\) places the abstract product-manifold point in the ambient matrix space, while \(\varphi\) gives local coordinates for analysis on the manifold.
\end{remark}

The physical feasible manifold is the embedded image
\vspace{-0.5ex}
\begin{equation}
    \Mgc\!
    \!=\!
    \iota(\Mhat)
    \!=\!
    \left\{
    \blkdiag(\bTheta_1,\ldots,\!\bTheta_G)\!:\!
    \bTheta_{\!g}\!\!\in\!\Us(R_G)\!
    \right\}\!\!.
    \label{eq:physical_manifold}
\end{equation}

The inverse of \(\iota\) on this image reads the diagonal blocks in \eqref{eq:group_scattering_block} for \(g=1,\ldots,G\).
Therefore \(\Mhat\) and \(\Mgc\) are diffeomorphic, although their elements are tuples and matrices, respectively.
Finally, the corresponding sum-rate maximization problem can be reformulated as
\vspace{-0.5ex}
\begin{equation}
\vspace{-0.5ex}
    \underset{\bTheta \in \Mgc}{\mathrm{maximize}} \quad \eta(\bTheta),
    \label{eq:group_rate_optimization}
\end{equation}
where \(\eta(\bTheta)\) is evaluated using the equivalent channel in \eqref{eq:group_equivalent_channel}. 

Although the feasible set has a product structure across the \(G\) groups, the objective function does not decompose into a sum of \(G\) independent objectives.
In particular, every block \(\bTheta_g\) contributes to the same equivalent channel \(\mE\), so changing one block affects the value of the global rate and the gradients associated with the remaining blocks. 
Hence, the product structure is useful for deriving blockwise tangent projections and retractions, but does not imply that the \(G\) blocks can be optimized as  independent rate-maximization problems.

\section{Product Tangent Space, Projection, and Retraction}

Let \(\bTheta_g=\mQ_g\mQ_g^T\) be a Takagi factorization of the \(g\)th block.
Santamaria \emph{et al.} show that the tangent space of \(\Us(R_G)\) at \(\bTheta_g\) is \cite{Santamaria2026SymmetricUnitary}
\begin{equation}
    \mathrm{T}_{\bTheta_{\!g}}\Us(R_G)
    \!=\!\!
    \left\{
    j\mQ_g\mR_g\mQ_g^T\!:\!
    \mR_g\!\in\!\R^{R_G\times R_G},\; \!\!
    \mR_g\!=\!\mR_g^T
    \right\}\!\!.
    \label{eq:factor_tangent}
\end{equation}

\begin{proposition}[Product tangent space and metric]
At \(\widehat{\bTheta}=(\bTheta_1,\ldots,\bTheta_G)\), the tangent space is
\begin{equation}
    \mathrm{T}_{\widehat{\bTheta}}\Mhat
    =
    \mathrm{T}_{\bTheta_1}\Us(R_G)
    \oplus\cdots\oplus
    \mathrm{T}_{\bTheta_G}\Us(R_G).
    \label{eq:product_tangent}
\end{equation}

{Here, $\oplus$ denotes the direct sum: a tangent vector on the product
manifold is a tuple of blockwise tangent directions, one from each group \cite{ConradTangentProducts}.
}

For tangent vectors \(\widehat{\bxi}=(\bxi_1,\ldots,\bxi_G)\) and \(\bzetahat=(\bm{\zeta}_1,\ldots,\bm{\zeta}_G)\), the product metric is
\begin{equation}
    \langle\widehat{\bxi},\bzetahat\rangle_{\widehat{\bTheta}}
    =
    \sum_{g=1}^{G}
    \Re\!\left\{\tr(\bxi_g^H\bm{\zeta}_g)\right\},
    \label{eq:product_metric}
\end{equation}

\end{proposition}
\begin{proof}
A smooth curve on \(\Mhat\) has the componentwise form \(c(t)=(c_1(t),\ldots,c_G(t))\), where \(c_g(t)\in\Us(R_G)\).
Differentiating at \(t=0\) gives \(\dot c(0)=(\dot c_1(0),\ldots,\dot c_G(0))\), which proves \eqref{eq:product_tangent}, while \eqref{eq:product_metric} is the standard product Riemannian metric obtained by summing the factor inner products \cite{Absil2008Optimization,Boumal2023SmoothManifolds}.
\end{proof}

Under the embedding \(\iota\), a tangent vector becomes the block-diagonal ambient matrix
\begin{equation}
    D\iota_{\widehat{\bTheta}}[\widehat{\bxi}]
    =
    \blkdiag(\bxi_1,\ldots,\bxi_G).
    \label{eq:embedded_tangent}
\end{equation}

Thus the tangent space respects the same group sparsity pattern as the feasible scattering matrix.

For an arbitrary ambient block direction \(\mJ_g\in\C^{R_G\times R_G}\),
which will later be chosen as the \(g\)th block of the Euclidean
rate-gradient matrix, the tangent projection map \(\Tmap_{\bTheta_g}:\C^{R_G\times R_G}\rightarrow \mathrm{T}_{\bTheta_g}\Us(R_G)\) is \cite{Santamaria2026SymmetricUnitary}
\vspace{-0.5ex}
\begin{gather} 
    \mR_g
    =
    \Im\!\left\{
    \mQ_g^H
    \frac{\mJ_g+\mJ_g^T}{2}
    \mQ_g^*
    \right\},
    \label{eq:projection_parameter}
    \\
    \Tmap_{\bTheta_g}(\mJ_g)
    =
    j\mQ_g\mR_g\mQ_g^T .
    \label{eq:factor_projection}
\vspace{-0.5ex}
\end{gather}

The product tangent projection map is obtained by applying \eqref{eq:factor_projection} to each block, as
\vspace{-0.5ex}
\begin{equation}
    \Tmap_{\widehat{\bTheta}}
    (\mJ_1,\ldots,\mJ_G)
    =
    \left(\Tmap_{\bTheta_g}(\mJ_g)\right)_{g=1}^{G}.
    \label{eq:product_projection}
\vspace{-0.5ex}
\end{equation}

Let \(\Pi_{\Us(R_G)}(\mA)=\mV\mV^T\) when \(\mA=\mV\bm{\Sigma}\mV^T\) is a Takagi factorization of a symmetric matrix.
The factor retraction is
\vspace{-0.5ex}
\begin{equation}
    \mathcal{R}_{\bTheta_g}(\bxi_g)
    =
    \Pi_{\Us(R_G)}(\bTheta_g+\bxi_g),
    \label{eq:factor_retraction}
\vspace{-0.5ex}
\end{equation}
which is the projection-like Takagi retraction used for \(\Us(R_G)\) \cite{Santamaria2026SymmetricUnitary,AbsilMalick2012ProjectionRetractions}.
The product retraction is
\vspace{-0.5ex}
\begin{equation}
\mathcal{R}_{\widehat{\bTheta}}(\widehat{\bxi})
    =
    \left(
    \mathcal{R}_{\bTheta_1}(\bxi_1),\ldots,
    \mathcal{R}_{\bTheta_G}(\bxi_G)
    \right).
    \label{eq:product_retraction}
\vspace{-0.5ex}
\end{equation}

Since each factor map satisfies the two retraction axioms, their product also satisfies them.

\vspace{-1ex}

{\section{Rate Gradient and Product MO-PO Algorithm}

Let
\begin{equation}
    \rho =
    \frac{P}{N_t\sigma^2},
    \label{eq:rho_definition}
\end{equation}
and
\begin{equation}
    \mK =
    \mI_{N_r}
    +
    \rho \mE\mE^H .
    \label{eq:rate_matrix}
\end{equation}

Using Wirtinger calculus, \(\bTheta\) and \(\bTheta^{\star}\) are treated as
independent variables. For the real-valued objective \(\eta\), the
Euclidean ascent direction is taken with respect to \(\bTheta^{\star}\). The
corresponding Euclidean gradient is \cite{Santamaria2026SymmetricUnitary,PalomarVerdu2006Gradient}
\begin{equation}
    \mJ
    \triangleq
    \nabla_{\bTheta}\eta \triangleq 2\frac{\partial \eta}{\partial \bTheta^{\star}}
    =
    \frac{2\rho}{\ln 2}
    \mHrx^H
    \mK^{-1}
    \mE
    \mHtx^H .
    \label{eq:full_gradient}
\end{equation}

\begin{algorithm}[t]
\caption{Product \ac{MO}-\ac{PO} for group-connected reciprocal \ac{BD-RIS}}
\label{alg:product_mopo}
\begin{algorithmic}[1]
\REQUIRE \(\mHd,\mHrx,\mHtx\), \(G\), \(R_G\), \(P\), \(\sigma^2\), tolerance \(\epsilon\)
\STATE Initialize \(\bTheta_g^{(0)}\in\Us(R_G)\), \(g=1,\ldots,G\)
\STATE Form \(\mE\) from \eqref{eq:group_equivalent_channel} and compute  {\(\eta^{(0)}\)} 
\FOR{\(r=0,1,\ldots\)}
    \FOR{\(g=1,\ldots,G\)}
        \STATE Compute \(\mK\) from \eqref{eq:rate_matrix} and \(\mJ_g\) from \eqref{eq:block_gradient}
        \STATE Compute a Takagi factor \(\bTheta_g=\mQ_g\mQ_g^T\)
        \STATE Compute \(\Tmap_{\bTheta_g}(\mJ_g)\) by \eqref{eq:projection_parameter} and \eqref{eq:factor_projection}
        \STATE Compute \(\mQ_{r,g}\) from \eqref{eq:projected_eigendecomposition}
        \STATE Optimize the phases in \eqref{eq:block_phase_parameterization} using the one-phase model \eqref{eq:one_phase_channel}
        \STATE Replace \(\bTheta_g\) by the updated symmetric-unitary block and recompute \(\mE\)
    \ENDFOR
    \STATE Compute the updated rate \(\eta^{(r+1)}\)
    \IF{\(|\eta^{(r+1)}-\eta^{(r)}|<\epsilon\)}
        \STATE \textbf{break}
    \ENDIF
\ENDFOR
\ENSURE \(\bTheta=\blkdiag(\bTheta_1,\ldots,\bTheta_G)\)
\end{algorithmic}
\end{algorithm}

Only the diagonal group blocks of \(\mJ\) can pair with feasible embedded
tangent directions, so that the \(g\)th block gradient is
\begin{equation}
    \mJ_g
    =
    \mJ_{[\Ig,\Ig]}
    =
    \frac{2\rho}{\ln 2}
    \mHrxg^H
    \mK^{-1}
    \mE
    \mHtxg^H .
    \label{eq:block_gradient}
\end{equation}

The positive scalar \(2\rho/\ln 2\) may be omitted when only the ascent
direction is needed, but it is retained here to state the exact Euclidean
gradient.}




The block \(\mJ_g\) depends on every group through \(\mE\), even though its tangent projection is performed on the \(g\)th factor only.
In turn, the \ac{PO} step follows the geodesic parameterization of \(\Us(R_G)\).
After projecting \(\mJ_g\), write the tangent direction as \(j\mQ_g\mR_g\mQ_g^T\) and compute the eigendecomposition
\vspace{-0.5ex}
\begin{equation}
    j\mR_g
    =
    \mV_g
    \diag(j\theta_{g,1},\ldots,j\theta_{g,R_G})
    \mV_g^T.
    \label{eq:projected_eigendecomposition}
\vspace{-0.5ex}
\end{equation}

{where \(\mV_g\) is real orthogonal and \(\theta_{g,k}\in\R\), and} \(\mQ_{r,g}=\mQ_g\mV_g\), {such that} the candidate block is
\vspace{-0.5ex}
\begin{equation}
    \bTheta_g(\bm{\phi}_g)
    =
    \mQ_{r,g}
    \diag(e^{j\phi_{g,1}},\ldots,e^{j\phi_{g,R_G}})
    \mQ_{r,g}^T.
    \label{eq:block_phase_parameterization}
\vspace{-0.5ex}
\end{equation}

For a trial value of phase \(k\) in group \(g\), with all other variables fixed, the equivalent channel can be written as
\begin{equation}
    \mE
    =
    \mS
    +
    e^{j\phi_{g,k}}\mathbf{h}_{\mathrm{RX},g,k}\mathbf{h}_{\mathrm{TX},g,k}^T ,
    \label{eq:one_phase_channel}
\end{equation}
{where, \(\mathbf q_{r,g,k}\in\C^{R_G}\) is the \(k\)th column of \(\mQ_{r,g}\), \(\mathbf h_{\mathrm{RX},g,k}=\mHrxg\mathbf q_{r,g,k}\in\C^{N_r}\), \(\mathbf h_{\mathrm{TX},g,k}=\mHtxg^T\mathbf q_{r,g,k}\in\C^{N_t}\), and}
\vspace{-1ex}
{
\begin{equation}
\vspace{-1ex}
\mS
=
\mHd
+ \hspace{-2ex}
\sum_{\ell=1, \ell\neq g}^{G} \hspace{-1ex}
\mathbf H_{\mathrm{RX},\ell}\bTheta_{\ell}\mathbf H_{\mathrm{TX},\ell}
+ \!\!\!\!\!\!
\sum_{m=1, m\neq k}^{R_G} \hspace{-2ex}
e^{j\phi_{g,m}}
\mathbf h_{\mathrm{RX},g,m}\mathbf h_{\mathrm{TX},g,m}^{T}.
\end{equation}}

{For each group update, \(\mQ_{r,g}\) is fixed and \(\phi_{g,1},\ldots,\phi_{g,R_G}\) are updated sequentially, each by globally maximizing \(\eta\) under \eqref{eq:one_phase_channel} using the latest values of the other phases.}

This is the same rank-one phase subproblem as in the fully connected \ac{MO}-\ac{PO} method \cite{Santamaria2026SymmetricUnitary,Santamaria2026ICASSP}, with the  dominant manifold operations in one outer sweep given by blockwise Takagi factorizations, eigendecompositions, and matrix multiplications.

{\paragraph{Convergence remark}
Each exact scalar phase update cannot decrease \(\eta\); compactness of \(\Mhat\) and continuity of \(\eta\) therefore guarantee convergence of the achieved-rate sequence.
If the scalar maximizers are unique and the factorization choices defining \(\mQ_{r,g}\) vary continuously, the cyclic coordinate-ascent conditions in \cite[Sec.~III-B]{Santamaria2026SymmetricUnitary} ensure that no phase can improve \(\eta\) at the limit of any convergent subsequence of iterates.
Since the curve \(\phi_{g,k}=t\theta_{g,k}\) starts in direction \(\Tmap_{\bTheta_g}(\mJ_g)\), a nonzero projected gradient would improve at least one phase; thus every such limit is first-order Riemannian stationary, although global optimality and convergence to a single point are not guaranteed.}

{For fixed antenna dimensions, the dominant manifold-geometry operations in one outer sweep are the blockwise Takagi factorizations and eigendecompositions, whose complexity scales as}
\begin{equation}
    \sum_{g=1}^{G}O(R_G^3).
    \label{eq:block_complexity}
\end{equation}

Since \(R=G R_G\), this becomes \(O(RR_G^2)\), compared with \(O(R^3)\) for a fully connected update on one \(R\times R\) block.
The group size therefore controls a direct performance-complexity tradeoff.

\begin{remark}[Per-group low-rank modification]
\label{rem:low_rank_modification}
The global low-rank construction in \cite{Santamaria2026SymmetricUnitary} uses a basis \(\mU_Z\) for the column space of \(\mZ=[\mHrx^H,\;\mHtx^*]\) and forms \(\bTheta_{\mathrm{lr}}=\mU_Z\widetilde{\bTheta}\mU_Z^T\).
This matrix is generally dense and therefore does not preserve the group-connected block structure.
A compatible modification of Algorithm~\ref{alg:product_mopo} would replace lines 6--10 for group \(g\) by a reduced problem built from
\begin{equation}
    \mZ_g=
    \left[
    \mHrxg^H,\;
    \mHtxg^*
    \right],
    \label{eq:per_group_low_rank_space}
\end{equation}
optimize a smaller symmetric-unitary factor in the corresponding subspace, and complete the result inside \(\Us(R_G)\).
This modification preserves the block-diagonal architecture, but its optimality is not claimed here.
\end{remark}

\section{Simulation Results}

We evaluate the achievable rate for \(N_t=N_r=4\) antennas and \(P_{\max}=0\;\mathrm{dBm}\).
The noise power is computed for a \(20\;\mathrm{MHz}\) bandwidth and a \(0\;\mathrm{dB}\) receiver noise figure.
The \ac{BD-RIS}-assisted links \(\mHrx\) and \(\mHtx\) follow a Rician fading model with factor \(K_{\mathrm{RIS}}=3\), while the direct channel \(\mHd\) uses \(K_{\mathrm{d}}=0\), corresponding to Rayleigh fading\footnote{The remaining channel parameters follow \cite[Sec.~IV-C]{Santamaria2026SymmetricUnitary}.}.

For each value of \(R\), the reported rate is averaged over \(100\) independent channel realizations.
Fig.~\ref{fig:rate_vs_elements} compares the resulting rate as the number of reconfigurable elements increases.
For clarity, different markers distinguish the architecture types, and such notation is used throughout the paper.

\begin{figure}[H]
    \centering
    \includegraphics[width=\columnwidth]{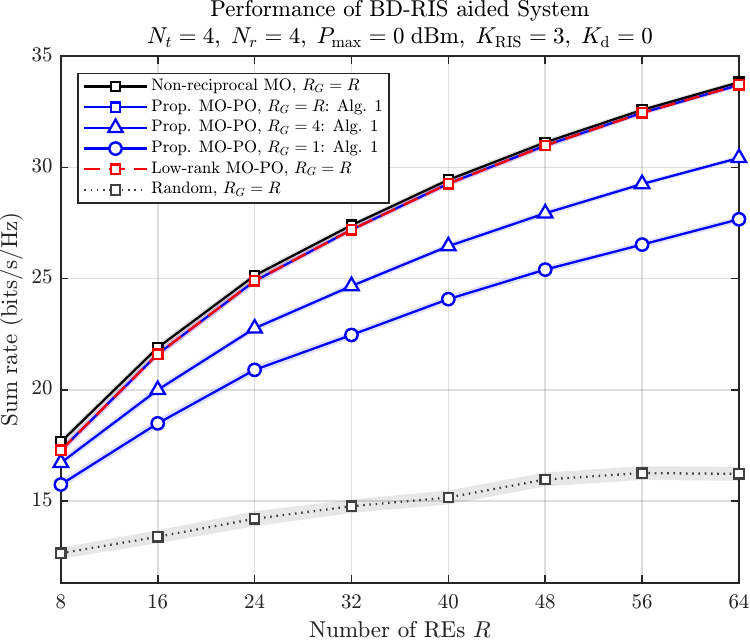}
    \caption{Achievable rate performance versus the number of reconfigurable elements (REs) for the considered reciprocal {and non-reciprocal} \ac{BD-RIS} architectures. {Recall, $R_G = 1$ corresponds to the conventional (diagonal) \ac{RIS}; furthermore, the shaded bands show 95\% confidence intervals of the Monte Carlo sample mean over 100 independent channel realizations.}}
    \label{fig:rate_vs_elements}
\vspace{-2ex}
\end{figure}

The results preserve the expected ordering.
Increasing \(R_G\) increases the number of symmetric-unitary degrees of freedom in \eqref{eq:product_dimension}, and the rate improves accordingly.
The full-block product \ac{MO}-\ac{PO} and the low-rank fully connected reference almost coincide in this experiment, which is consistent with the fully connected low-rank result in \cite{Santamaria2026SymmetricUnitary}.
For the group-connected architecture, however, the low-rank reduction must be made compatible with the block structure as discussed in Remark~\ref{rem:low_rank_modification}.

{Fig.~\ref{fig:rate_vs_elements} also has several benchmark alternatives for comparison.
This includes (in grey) a curve corresponding to the case when the scattering matrix is randomly generated as a unitary symmetric matrix, which serves as a baseline, since such a ``solution'' only satisfies the practical constraints of the problem, but not the optimization objective.
Next, as a best feasible case reference, the results corresponding to a non-reciprocal \ac{BD-RIS} scattering matrix design obtained by maximizing \eqref{eq:rate_objective} over \(\mathrm{U}(R)=\mathrm{St}_{\C}(R,R)\) using Manopt \cite{Boumal2014Manopt} are also shown (in black).
Since this case does not incorporate a reciprocity constraint, such that \(\Us(R)\subset\mathrm{U}(R)\), the results obtained over \(\mathrm{U}(R)\) represent an upper-bound to the reciprocal solution.

Remarkably, it is found that the reciprocal and non-reciprocal curves nearly coincide, which indicates that the reciprocal design has no noticeable loss compared to the non-reciprocal counterpart under the considered model.
This result is particularly motivating since gains of non-reciprocal designs over reciprocal ones have been reported, e.g., in interference-coupled settings \cite{Li2024NonReciprocalBDRIS}.
Further evaluation under other system models is left for a journal extension of this work.

Finally, the figure also includes the low-rank solution from \cite{Santamaria2026SymmetricUnitary}, limited to the fully connected architecture, which also coincides with our design under that condition.
Taken together with the fact that our proposed product manifold method is, on the other hand, generalized to arbitrary group-connected settings, the results corroborate that our scheme generalizes that of \cite{Santamaria2026SymmetricUnitary}, without any penalty in the fully connected case.}


\vspace{-0.25ex}
\section{Conclusion}

We proposed a sum-rate maximization framework for reciprocal group-connected \ac{BD-RIS}-assisted \ac{MIMO} systems, modeling the feasible space as a product manifold formed from one symmetric-unitary admissible set per group. The smoothness of each factor was established via a Takagi characterization, and the resulting tuple space was equipped with the product topology and product charts to form a smooth product manifold, kept distinct from its block-diagonal physical embedding in the ambient matrix space. This separation enables blockwise tangent projections and Takagi retractions consistent with the group-connected hardware, while the rate objective remains globally coupled through the common equivalent channel \(\mE\). The framework includes the single-connected and fully connected architectures as special cases, with the group size governing a direct performance-complexity tradeoff confirmed by our simulations.

\vspace{-0.5ex}
\appendices
\section{Detailed Product-Manifold Argument}
\label{app:smooth}

Let \(\mathcal{M}_g=\Us(R_G)\) be equipped with the smooth structure from Proposition~2.
The product topology on \(\Mhat=\prod_g\mathcal{M}_g\) is generated by products of open factor neighborhoods.
For \(\widehat{\bTheta}\in\Mhat\), the product chart in \eqref{eq:product_chart} maps an open neighborhood of \(\widehat{\bTheta}\) to an open subset of \(\R^{Gd_g}\).
If another product chart \(\psi\) is chosen, the coordinate transition is
\begin{equation}
    \varphi\circ\psi^{-1}
    =
    \left(
    \varphi_1\circ\psi_1^{-1},
    \ldots,
    \varphi_G\circ\psi_G^{-1}
    \right),
    \label{eq:product_transition}
\end{equation}
which is smooth because every factor transition is smooth.
Thus \(\Mhat\) is a smooth manifold of dimension \(Gd_g\).

The embedding \(\iota\) in \eqref{eq:physical_embedding} is smooth because each matrix entry of \(\blkdiag(\bTheta_1,\ldots,\bTheta_G)\) is either an entry of one factor block or zero.
It is injective, and its inverse on \(\Mgc\) reads the diagonal blocks in \eqref{eq:group_scattering_block}.
Both maps are smooth in the induced matrix coordinates, so \(\Mgc\) is diffeomorphic to \(\Mhat\).



\end{document}